\documentclass{llncs}
\usepackage[utf8]{inputenc}
\usepackage{amsmath,amsfonts,amssymb}
\usepackage{bbold,hyperref}
  \usepackage{listings}             

  \usepackage{algorithm}
  \usepackage{algorithmic}
      \floatname{algorithm}{Algoritmo}

    \usepackage{graphicx} 
    \usepackage{float}        
    \usepackage{eso-pic}    
    \usepackage{sidecap}      
    \usepackage{wrapfig}      
    \usepackage[compatibility=false]{caption}      
    \usepackage{subcaption}
    \usepackage{tikz}
    \usepackage{tikz-3dplot}
    \usetikzlibrary{arrows.meta}
    \usetikzlibrary{decorations.markings}

    \usepackage{color}          
    \usepackage{multicol}       
    \usepackage{enumerate}      
\usepackage{placeins}
\newcommand{\ldarrow}{\leftarrow\!\!\!\!\downarrow}

\newcommand{\MNpic}[9]{
\raisebox{-3.0 ex}{
\begin{tikzpicture}[thick,scale=0.4, every node/.style={scale=1.0}]
\draw  (-3,5) rectangle (0,2);
\draw [help lines, step=1cm] (-3,2) grid (0,5);
\node at (-2.5,4.5) {#1};
\node at (-1.5,4.5) {#2};
\node at (-0.5,4.5) {#3};
\node at (-2.5,3.5) {#4};
\node at (-1.5,3.5) {#5};
\node at (-0.5,3.5) {#6};
\node at (-2.5,2.5) {#7};
\node at (-1.5,2.5) {#8};
\node at (-0.5,2.5) {#9};
\end{tikzpicture}
}}
\newcommand{\MNpicf}[9]{
\raisebox{-2.750 ex}{
\begin{tikzpicture}[thick,scale=0.3, every node/.style={scale=0.85}]
\draw  (-3,5) rectangle (0,2);
\draw [help lines, step=1cm] (-3,2) grid (0,5);
\node at (-2.5,4.5) {#1};
\node at (-1.5,4.5) {#2};
\node at (-0.5,4.5) {#3};
\node at (-2.5,3.5) {#4};
\node at (-1.5,3.5) {#5};
\node at (-0.5,3.5) {#6};
\node at (-2.5,2.5) {#7};
\node at (-1.5,2.5) {#8};
\node at (-0.5,2.5) {#9};
\end{tikzpicture}
}}

\newcommand{\VNpic}[5]{
\raisebox{-3.0 ex}{
\begin{tikzpicture}[thick,scale=0.4, every node/.style={scale=1.0}]
\draw  (-1.5,3.5) rectangle (-0.5,0.5);
\draw  (-2.5,2.5) rectangle (0.5,1.5);
\node at (-1,3) {#2};
\node at (-1,2) {#1};
\node at (-1,1) {#4};
\node at (-2,2) {#5};
\node at (0,2)  {#3};
\end{tikzpicture}
}}
\newcommand{\LNpic}[3]{
\raisebox{-2.0 ex}{
\begin{tikzpicture}[thick,scale=0.5, every node/.style={scale=0.75}]
\draw  (-1.5,3.5) rectangle (-0.5,1.5);
\draw  (-1.5,2.5) rectangle (0.5,1.5);
\node at (-1,2) {$#1$};
\node at (-1,3) {$#2$};
\node at (0,2)  {$#3$};
\end{tikzpicture}
}}
\newcommand{\rbsquare}{
\raisebox{-0 ex}{
\begin{tikzpicture}[thick,scale=0.475, every node/.style={scale=0.75}]
\fill (-4.5,6.5) -- (-4.5,7) -- (-4,7) -- (-4.5,6.5);
\fill[red] (-4,6.5) -- (-4,7) -- (-4.5,6.5) -- (-4,6.5);
\end{tikzpicture}
}}

\newcommand{\cubo}[3]{
    \draw[red,fill=yellow] (#1,#2,#3) -- ++(-\cubex,0,0) -- ++(0,-\cubey,0) -- ++(\cubex,0,0) -- cycle;
    \draw[red,fill=yellow] (#1,#2,#3) -- ++(0,0,-\cubez) -- ++(0,-\cubey,0) -- ++(0,0,\cubez) -- cycle;
    \draw[red,fill=yellow] (#1,#2,#3) -- ++(-\cubex,0,0) -- ++(0,0,-\cubez) -- ++(\cubex,0,0) -- cycle;
}

\newcounter{TODOcounter}

\newcommand\Z{\mathbb{Z}}
\newcommand{\FCA}{\ensuremath{\mathbf{FCA}}}
\newcommand{\VN}{\ensuremath{\mathrm{VN}}}
\newcommand{\MN}{\ensuremath{\mathrm{MN}}}
\newcommand{\LN}{\ensuremath{\mathrm{LN}}}

\newcommand\myparagraph[1]{\par\textbf{#1.}}

\title{Universality in Freezing Cellular Automata}

\newcommand\LIFO{Universit\'{e} Orl\'{e}ans, LIFO EA 4022, FR-45067 Orl\'{e}ans, France}
\author{Florent Becker\inst{1}\and Diego Maldonado\inst{1}\and Nicolas Ollinger\inst{1}\and Guillaume Theyssier\inst{2}}
\institute{\LIFO\and Institut de Mathématiques de Marseille (Université Aix Marseille, CNRS, Centrale Marseille), France, \email{guillaume.theyssier@cnrs.fr}}

\authorrunning{F. Becker, D. Maldonado, N. Ollinger and G. Theyssier}

\begin{document}
\maketitle

\begin{abstract}
  Cellular Automata have been used since their introduction as a discrete tool of modelization. In many of the physical processes one may modelize thus (such as bootstrap percolation, forest fire or epidemic propagation models, life without death, etc), each local change is irreversible. The class of \emph{freezing} Cellular Automata (FCA) captures this feature. In a freezing cellular automaton the states are ordered and the cells can only decrease their state according to this ``freezing-order''.

  We investigate the dynamics of such systems through the questions of simulation and universality in this class: is there a Freezing Cellular Automaton (FCA) that is able to simulate any Freezing Cellular Automata, i.e. an intrinsically universal FCA? We show that the answer to that question is sensitive to both the number of changes cells are allowed to make, and geometric features of the space. In dimension 1, there is no universal FCA. In dimension 2, if either the number of changes is at least 2, or the neighborhood is Moore, then there are universal FCA. On the other hand, there is no universal FCA with one change and Von Neumann neighborhood. We also show that monotonicity of the local rule with respect to the freezing-order (a common feature of bootstrap percolation) is also an obstacle to universality. 
\end{abstract}

\section{Introduction}
Cellular automata (CA for short) are discrete dynamical systems at the crossroad of several reasearch fields and points of view (dynamical systems theory, computer science, physical modeling, etc). In the pioneering works impulsed by J. von Neumann and S. Ulam in the 50-60s, when cellular automata were formally defined for the first time, two important themes were already present: universality \cite{vonneumann,holland,thatcher} and growth dynamics \cite{ulam}. Since then, these themes have received a considerable attention in the literature. Concerning universality, production of examples \cite{banks,OllingerR11} was accompanied by progresses on the formalization and the theoretical analysis of the concept \cite{OllingerUniv}, in particular with the emergence of intrinsic simulations and universality \cite{bulking1,bulking2}. Growing dynamics in cellular automata were also much studied, mostly through (classes of) examples with different points of view \cite{Gravner98,GriMoo96,Fuentes,bollobas2015}. More recently, substantial work have been published on models of self-assembly tilings, most of which can be seen as a particular non-deterministic 2D CA where structures grow from a seed. Interestingly, the question of intrinsic universality was particularly studied in that case \cite{DotyLPSSW12,MeunierPSTWW14}.

A common feature of all these examples is that only a bounded number of changes per cell can occur during the evolution. To our knowledge, the first time that the class of CAs with that feature was considered as a whole is in \cite{vollmar81} with a point of view of language recognition. More recently the notion of \emph{freezing} CA was introduced in \cite{GolOlThey15}, which captures essentially the same idea with an explicit order on states, and a systematic study of this class (dynamics, predictability, complexity) was started. In particular it was established that the class is Turing universal (even in dimension 1).

In this paper, we study intrinsic universality in freezing CA as a first step to understand universality in growth dynamics in general. Our central result is the construction of such intrinsically universal freezing CA: it shows that the class of freezing CA is a natural computational model with maximally complex elements which can be thought of as machines that can be `programmed' to produce any behavior of the class. Moreover, the universal CA that we construct are surprisingly small (5 states, see Section~\ref{sec:2d_2change_VN}) which is in strong contrast with the complicated construction known to obtain intrinsic universality for the classical self-assembly aTAM model~\cite{DotyLPSSW12}. Our contribution also lays in the negative results we prove (Theorems~\ref{teo:1Dsim}, \ref{teo: nfu1D} and \ref{teo: nfum}): interpreting them as necessary conditions to achieve universality for freezing CA, we obtain a clear landscape of the fundamental computational or dynamical features of this class. 

The paper is organized as follows. In Section~\ref{sec:defs} we define the main concepts (freezing CA, and intrinsic simulation) and prove that the use of context-free simulation cannot lead to universality. Section~\ref{sec:UFCA} gives a general construction scheme to obtain universal freezing CA giving three positive results in three different settings depending on the dimension, the neighborhood and the maximum number of state changes per cell. In section~\ref{sec:obstacles}, we show several obstacles to the existence of universal freezing CA: dimension $1$, 1 change per cell with von Neumann neighborhood in 2D, and monotonicity. 

\section{Definitions}
\label{sec:defs}

\newcommand\cadef{\mathcal{CA}}
\begin{definition}
  A \emph{cellular automaton} $F$ of dimension $d$ and state set $Q$ is a tuple $F = (d, Q, N, f)$, where $d$, its \emph{dimension} is an integer, $Q$, its set of \emph{states} is a finite set, $N \subset \mathbb{Z}^d$ is its finite \emph{neighborhood}, and $f: Q^N \rightarrow Q$ is its \emph{local function}.

  It induces a global function, which we also note $F$, acting on the
set of configurations $Q^{\mathbb{Z}^{d}}$ as follows:
$$
\forall c\in Q^{\mathbb{Z}^{d}},\forall z\in \mathbb{Z}^{d},\ F(c)_{z}=f(c|_{z+N})
$$

\end{definition}

Let $e_1, \ldots, e_d$ be the canonical basis of $\Z^d$; $\VN_d=\{\vec{0},e_1, \ldots, e_d\}$ is the \emph{von Neumann neighborhood}. We also use the following neighborhoods in dimension 2: $\MN=\{(0,0),(\pm1,0),(0,\pm1),(\pm1,\pm1),(\mp1,\pm1)\}$ is the \emph{Moore Neighborhood}; $\LN=\{(0,0),(1,0),(0,1)\}$ is the \emph{L-neighborhood}. 

In many cellular automata from the literature, there is a global bound on the number of times a cell can change: they are \emph{bounded-change} cellular automata. This property is found in most of the Cellular Automata considered in bootsrap percolation, as well as in other well-known examples such as 'Life without death' \cite{GriMoo96} and various models of propagation phenomena like in~\cite{Fuentes}. We say that a CA is \emph{$k$-change} if any cell in any orbit changes at most $k$ times of state.

 Moreover, in all those examples, the bound is defined through an explicit order on states. Such an automaton is a \emph{freezing cellular automaton}; they were introduced in \cite{GolOlThey15}. This \emph{freezing-order} on state can also be used to define interesting subclasses (see Section~\ref{sec:monotone}).

\begin{definition}[Freezing Cellular Automaton]
  A CA $F$  is a \emph{$\prec$-freezing} CA, for some (partial) order $\prec$ on states, if $F(c)_{z}\prec c_{z}$ for any configuration $c$  and any cell $z$. A CA is \emph{freezing} if it is $\prec$-freezing for some order.
\end{definition}
 Any freezing CA is $k$-change for some $k$ (at most the depth of its freezing order, but possibly less). For $V \subset \Z^d$, we note $\FCA_V$ for the class of $d$-dimensional freezing cellular automata with neighborhood $V$. Finally, we set $\FCA^d = \bigcup_{V \subset \Z^d} \FCA_V$ and omit $d$ when context makes it clear.


Intrinsic universality is defined through a notion of simulation between CA.
 Roughly speaking, $F$ simulates $G$ if there is an encoding of configurations of $G$ into configurations of $F$ such that one step of $G$ is simulated through this encoding by a fixed number of steps of $F$.

\begin{definition}[Simulation] Let ${T>0}$, and $B\subseteq\Z^d$ be a $d$-dimensional rectangular block, with size-vector $b \in \mathbb{Z}^d$. Let $C \subset \mathbb{Z}^d$ be a finite set, with $\vec0 \in C$.
  Let ${F=(d, Q, N, f)}$ and ${G=(d,Q', N' ,g)}$ be two $d$-dimensional cellular automata. $F$ \emph{simulates} $G$ with slowdown $T$, block $B$ and context $C$ if there is a coding map ${\phi:Q_G^C\rightarrow Q_F^B}$ such that the global map ${\bar\phi : Q_G^{\Z^d}\rightarrow Q_F^{\Z^d}}$ verifies:
  \begin{itemize}
  \item ${\bar\phi}$ is injective;
  \item \(\forall c \in Q_G^{\mathbb{Z^d}}: \bar{\phi}(G(c)) = F^T(\bar{\phi}(c))\).    
  \end{itemize}
   where $\bar{\phi}$ is defined by: for $z \in \mathbb{Z}^d, r \in B$, $\overline{\phi}(c)_{bz + r}=\phi(c|_{z+C})_{r}$ 
 \label{def:contextsensitivesimu}
\end{definition}

When $C = \{ \vec0 \}$, this definition corresponds with the classical notion of 'injective simulation', as in \cite{bulking1,bulking2} and we call it \emph{context-free}.

\begin{figure}[hb]
    \begin{subfigure}[b]{0.475\textwidth}
        \centering
\begin{tikzpicture}[thick,scale=0.4, every node/.style={scale=0.45}]

\draw[thick]  (-3.5,3.5) rectangle (-1.5,-2.5);
\draw[thick]  (-1.5,3.5) rectangle (0.5,-2.5);
\draw[thick]  (0.5,3.5) rectangle (2.5,-2.5);
\draw[thick]  (-3.5,1.5) rectangle (2.5,3.5);
\draw[thick]  (-3.5,-0.5) rectangle (2.5,1.5);

\node[style={scale=1.5}] at (-2.5,2.5) {$A$};
\node[style={scale=1.5}] at (-0.5,2.5) {$B$};
\node[style={scale=1.5}] at (1.5,2.5) {$C$};
\node[style={scale=1.5}] at (-2.5,0.5) {$D$};
\node[style={scale=1.5}] at (-0.5,0.5) {$E$};
\node[style={scale=1.5}] at (1.5,0.5) {$F$};
\node[style={scale=1.5}] at (-2.5,-1.5) {$G$};
\node[style={scale=1.5}] at (-0.5,-1.5) {$H$};
\node[style={scale=1.5}] at (1.5,-1.5) {$I$};

\draw  (-10,2) rectangle (-7,1);
\draw  (-10,1) rectangle (-7,0);
\draw  (-10,0) rectangle (-7,-1);
\draw  (-9,2) rectangle (-10,-1);
\draw  (-8,2) rectangle (-9,-1);

\node[style={scale=1.5}] at (-9.5,1.5) {$a$};
\node[style={scale=1.5}] at (-8.5,1.5) {$b$};
\node[style={scale=1.5}] at (-7.5,1.5) {$c$};
\node[style={scale=1.5}] at (-9.5,0.5) {$d$};
\node[style={scale=1.5}] at (-8.5,0.5) {$e$};
\node[style={scale=1.5}] at (-7.5,0.5) {$f$};
\node[style={scale=1.5}] at (-9.5,-0.5) {$g$};
\node[style={scale=1.5}] at (-8.5,-0.5) {$h$};
\node[style={scale=1.5}] at (-7.5,-0.5) {$i$};

\draw[very thick,->] (-6.5,0.5) -- (-4,0.5);
 
\draw[very thick,red] (-1.5,1.5) rectangle (0.5,-0.5);
\node[style={scale=1.65}] at (-5.25,1) {$\overline{\phi}$};
\draw[very thick,->,red]  plot[smooth, tension=.7] coordinates {(-8.5,0.75) (-5,2.5) (-0.5,0.75)};
\node[style={scale=1.65}] at (-5.25,3) {${\phi}$};
\draw[very thick,red]  (-9,1) rectangle (-8,0);
\end{tikzpicture}
        \caption{Context-free simulation, $C = \{\vec0\}$.}
    \end{subfigure}%
    \begin{subfigure}[b]{0.475\textwidth}
        \centering
\begin{tikzpicture}[thick,scale=0.4, every node/.style={scale=0.45}]

\draw[thick]  (-3.5,3.5) rectangle (-1.5,-2.5);
\draw[thick]  (-1.5,3.5) rectangle (0.5,-2.5);
\draw[thick]  (0.5,3.5) rectangle (2.5,-2.5);
\draw[thick]  (-3.5,1.5) rectangle (2.5,3.5);
\draw[thick]  (-3.5,-0.5) rectangle (2.5,1.5);
\draw (-3.5,2.5) -- (-2.5,3.5);
\draw (-3.5,0.5) -- (-0.5,3.5);
\draw (-3.5,-1.5) -- (1.5,3.5);

\draw (1.5,-2.5) -- (2.5,-1.5);
\draw (-3,3) -- (-3.5,3.5);
\draw (-2,2) -- (-1,1);
\draw (0,0) -- (1,-1);
\draw (2,-2) -- (2.5,-2.5);
\draw (-3,1) -- (-3.5,1.5);
\draw (-2,0) -- (-1,-1);
\draw (0,2) -- (1,1);
\draw (-1,3) -- (-1.5,3.5);
\draw (2,2) -- (2.5,1.5);
\draw (2.5,2.5) -- (-2.5,-2.5);
\draw (2.5,0.5) -- (-0.5,-2.5);
\draw (0,-2) -- (0.5,-2.5);
\draw (2,0) -- (2.5,-0.5);
\draw (1,3) -- (0.5,3.5);
\draw (-3,-1) -- (-3.5,-0.5);
\draw (-2,-2) -- (-1.5,-2.5);

\node[style={scale=1.5}] at (-2.5,2.5) {$A$};
\node[style={scale=1.5}] at (-0.5,2.5) {$B$};
\node[style={scale=1.5}] at (1.5,2.5) {$C$};
\node[style={scale=1.5}] at (-2.5,0.5) {$D$};
\node[style={scale=1.5}] at (-0.5,0.5) {$E$};
\node[style={scale=1.5}] at (1.5,0.5) {$F$};
\node[style={scale=1.5}] at (-2.5,-1.5) {$G$};
\node[style={scale=1.5}] at (-0.5,-1.5) {$H$};
\node[style={scale=1.5}] at (1.5,-1.5) {$I$};

\node at (-1.25,3) {$a$};
\node at (-3,1.25) {$a$};

\node at (0.75,3) {$b$};
\node at (-1,1.25) {$b$};
\node at (-1.75,2) {$b$};

\node at (0.25,2) {$c$};
\node at (1,1.25) {$c$};

\node at (-2,1.75) {$d$};
\node at (-1.25,1) {$d$};
\node at (-3,-0.75) {$d$};

\node at (0,1.75) {$e$};
\node at (0.75,1) {$e$};
\node at (-1.75,0) {$e$};
\node at (-1,-0.75) {$e$};

\node at (2,1.75) {$f$};
\node at (0.25,0) {$f$};
\node at (1,-0.75) {$f$};

\node at (-2,-0.25) {$g$};
\node at (-1.25,-1) {$g$};

\node at (0,-0.25) {$h$};
\node at (0.75,-1) {$h$};
\node at (-1.75,-2) {$h$};

\node at (2,-0.25) {$i$};
\node at (0.25,-2) {$i$};
\draw  (-10,2) rectangle (-7,1);
\draw  (-10,1) rectangle (-7,0);
\draw  (-10,0) rectangle (-7,-1);
\draw  (-9,2) rectangle (-10,-1);
\draw  (-8,2) rectangle (-9,-1);

\node[style={scale=1.5}] at (-9.5,1.5) {$a$};
\node[style={scale=1.5}] at (-8.5,1.5) {$b$};
\node[style={scale=1.5}] at (-7.5,1.5) {$c$};
\node[style={scale=1.5}] at (-9.5,0.5) {$d$};
\node[style={scale=1.5}] at (-8.5,0.5) {$e$};
\node[style={scale=1.5}] at (-7.5,0.5) {$f$};
\node[style={scale=1.5}] at (-9.5,-0.5) {$g$};
\node[style={scale=1.5}] at (-8.5,-0.5) {$h$};
\node[style={scale=1.5}] at (-7.5,-0.5) {$i$};

\draw[very thick,->] (-6.5,0.5) -- (-4,0.5);
\draw[very thick,red] (-9,2) -- (-9,1) -- (-10,1) -- (-10,0) -- (-9,0) -- (-9,-1) -- (-8,-1) -- (-8,0) -- (-7,0) -- (-7,1) -- (-8,1) -- (-8,2) -- (-9,2);
\draw[very thick,red] (-1.5,1.5) rectangle (0.5,-0.5);
\node[style={scale=1.65}] at (-5.25,1) {$\overline{\phi}$};
\draw[very thick,->,red]  plot[smooth, tension=.7] coordinates {(-8.5,0.75) (-5,2.5) (-0.5,0.75)};
\node[style={scale=1.65}] at (-5.25,3) {${\phi}$};
\end{tikzpicture}
        \caption{Context-sensitive simulation with context $C = \VN_2$.}
    \end{subfigure}%
    \caption{Coding of states into blocks for the two modes of simulation}
\end{figure}
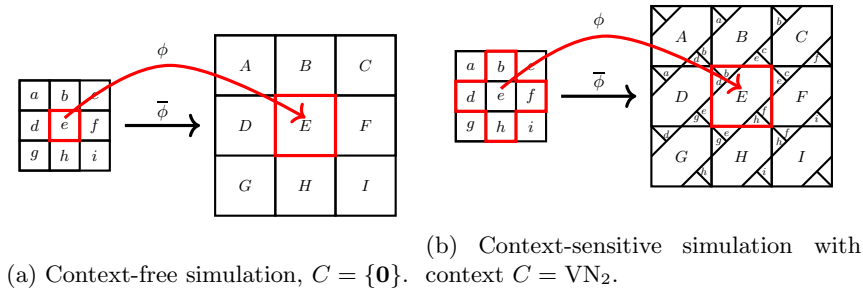

 In the context of freezing CAs, context-free universality is prevented by the irreversibility of any computation performed by a putative $U$ combined with the injectivity of the coding map, as witnessed by the following theorem. 

\begin{theorem}[No freezing context-free universality]
  \label{teo: context sim}
  Let $d \in \mathbb{N}$, there is no $F \in \FCA^d$ which is context-free $\FCA_{\VN_d}$-universal.
\end{theorem}

Context-sensitive simulation can get us over this hurdle as we show below; it is akin to the notion of conjugacy in symbolic dynamics \cite{LindMarcus}.

\section{Constructing Intrinsically Universal FCA}
\label{sec:UFCA}

We give a number of constructions for intrinsically universal freezing cellular automata. All of these exhibit the same running theme: if there is a means of crossing information asynchronously, then universality can be reached. This insight yields three constructions which are concrete implementations under various technical constraints of a common abstract construction.
\label{sec:abstract-sim}
The abstract construction can be described by: the structure of macro-cells, the mechanism to trigger state change in each macro-cell, and the wiring between neighboring macro-cells to ensure communication. At this abstract level we assume that there is a mean to cross wires without interference.
Another aspect of wiring is the necessity to put delays on some wires in order to keep synchronicity of information: it is a standard aspect of circuit encoding in CAs \cite{banks,OllingerUniv},  which we won't address in detail here but which can be dealt with by having wires make zigzag to adjust their length as desired.
The freezing condition imposes strong restrictions on the way we can code, transport and process information. We focus below on where our construction differs from the classical approach in general CAs.

\myparagraph{Wires are Fuses}
It is not possible to implement classical wires where bits of information travel freely without violating the freezing condition. In all of our constructions wires are actually fuses that can be used only once and they are usually implemented with two states: $1$ stays stable without presence of neighboring $0$s and $0$ propagates over neighboring $1$s. With that behavior our wires can be trees connecting various positions in such a way that a $0$ appearing at any position is broadcasted to the whole tree. A finite wire can either be uniformly in state ${b\in\{0,1\}}$ in which case all leaves 'agree' on the bit of information transported by it, or not uniform in which case information is incoherent between leaves. As it will become clear later, our constructions will use wires between adjacent blocks (or macro-cells) in the simulator CA and our encodings require that those wires are in a coherent state (uniformly ${b\in\{0,1\}}$): it is precisely in this aspect that we use the power of context sensitive simulations, because the content of a block (or macro-cell) cannot be fixed independently of its neighbors in that case. 

\myparagraph{State Codification}
In each macro-cell we must code in some way a (possibly very big) state that can change a (possibly very big) number of times: a classical binary encoding would violate the freezing condition so we actually use a unary coding. Given a finite set $S$ and a quasi order $(Q,\preceq)$, let $q_0 \preceq \ldots \preceq q_{|Q|-1}$ be a linearization of $\preceq$, and let $\iota(q_i) = i$.
Then let $Q_u = 1^*0^+ \cap \{0,1\}^{|Q|}$, and $\phi \in Q \rightarrow Q_u: q \mapsto 1^{\iota(q)}0^{|Q|-\iota(q)}$.
Note that for any $i < |Q|$ we have ${q \preceq q'\Leftrightarrow\phi(q)_i \leq \phi(q')_i}$ (where $\leq$ is the lexicographic order).
Since $\phi$ is a bijection, for any cellular automaton $F$ with state set $Q$, $\bar{\phi} \circ F \circ \bar{\phi}^{-1}$ is a cellular automaton isomorphic to $F$, with state set $Q_u$, which we call the \emph{unary representation} of $F$.
If $F$ is $\preceq$-freezing, then its unary representation is $\leq$-freezing.
We will use this unary encoding everywhere in the structure of our macro-cells: each state of a simulated CA $F$ will be represented by a collection of wires representing the bits of an element of $Q_u$ defined above. This encoding is coherent with the freezing property of the simulated CA because the fact that states can only decrease corresponds to the fact that the number of wires uniformly equal to $0$ increases.

\myparagraph{Neighborhood Matchers}
The fundamental basic block of our construction is a circuit that detects a fixed pattern in the neighborhood and outputs a bit of information saying: ``given this particular neighborhood pattern $w$, the new state of the macro-cell must be smaller than $l$''. Our unary encoding is adapted for this because the predicate ``smaller than $l$'' for a state translates into a condition on a single bit of an element of $Q_u$, that is to say a single wire in our concrete representation of states.
Without loss of generality we assume that $F=(\mathbb{Z}^2,Q_u,f,N)$ is a FCA with state in unary representation.
Take $L = |Q_u|$, and $m = |N|$.
For $l \in Q_u$, let $\{w_1^l,...,w_{K_l}^l\} = \bigcup_{s' \leq l}f^{-1}(s') = \{n \in (Q_u)^{N} | f(n) \leq l\}$. Take, for some state $l$, $w_k^l=({q}_1,...,{q}_m)\in \bigcup_{s' \leq l}f^{-1}(s')$ a fixed neighborhood with output smaller than $l$; each state $q_i$ is in $\{0,1\}^L$, so $w_k^l$ is a binary word in $\{0,1\}^{mL}$.
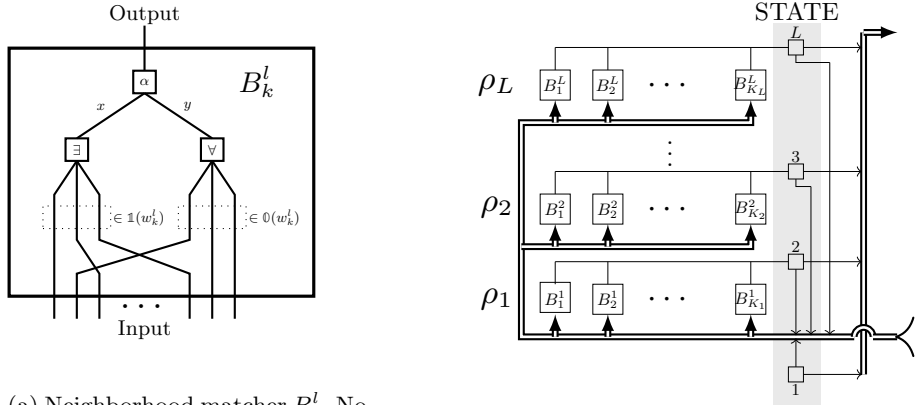
\begin{figure}[h]
  \begin{subfigure}{.4\textwidth}
\begin{tikzpicture}[thick,scale=0.3, every node/.style={scale=0.6}]

\draw[thick]  (-3.5,2.5) rectangle (-2.5,1.5);
\draw[thick]  (2.5,2.5) rectangle (3.5,1.5);
\draw (-4,-5.5) -- (-4,-3.5) -- (-4,-2) -- (-4,0) -- (-3,1.5);
\draw (-3,-5.5) -- (-3,-3.5) -- (2,-2)-- (2,0) -- (3,1.5);
\draw (-2,-5.5) -- (-2,-3.5) -- (-3,-2) -- (-3,0) -- (-3,1.5);
\draw (2,-5.5) -- (2,-3.5) -- (-2,-2) -- (-2,0) -- (-3,1.5);
\draw (3,-5.5) --(3,-3.5) -- (3,-2) -- (3,0) -- (3,1.5);
\draw (4,-5.5) -- (4,-3.5) -- (4,-2) -- (4,0) -- (3,1.5);
\draw[thick]  (-0.5,5.5) rectangle (0.5,4.5);
\draw (-3,2.5)  -- node[midway,above left] {$x$} (0,4.5);
\draw (3,2.5)  -- node[midway,above right] {$y$} (0,4.5);
\draw (0,5.5) -- (0,7.5);
\node at (-3,2) {$\exists$};
\node at (3,2) {$\forall$};
\node at (0,5) {$\alpha$};
\node[style={scale=1.5}] at (0,8) {Output};
\node[style={scale=2.5}] at (0,-5) {$\cdots$};
\draw[thin,dotted]  (-4.5,-0.5) rectangle (-1.5,-1.5);
\node at (-0.25,-1) {$\in \mathbb{1}(w_k^l)$};
\draw[thin,dotted]  (1.5,-0.5) rectangle (4.5,-1.5);
\node at (5.75,-1) {$\in \mathbb{0}(w_k^l)$};
\node[style={scale=1.5}]  at (0,-6) {Input};
\draw[very thick]  (-6,6.5) rectangle (7.5,-4.5);
\node[style={scale=2}] at (5,5) {$B_k^l$};
\end{tikzpicture}
    \caption{Neighborhood matcher $B_k^l$. Notations $\mathbb{0}(w) $ and $\mathbb{1}(w)$ stand for the set of all indexes $i$ s.t. $w_i=0$ and $w_i=1$ respectively. This block triggers a $0$ on the output wire exactly when the input is $w_k^l$. }
    \label{fig: checker}
  \end{subfigure}
  \hfill
  \begin{subfigure}{.5\textwidth}

\begin{tikzpicture}[scale=0.25,every node/.style={scale=0.7}]
\fill[gray!17]  (14,10) rectangle (16.5,-10.5);

\draw[line width=0.05](6.0,-3.9) rectangle (4.4,-5.5);
\draw[line width=0.05](3.2,-3.8) rectangle (1.6,-5.4);
\draw[line width=0.05](13.6,-3.9) rectangle (12.0,-5.5);
\draw[line width=0.1,double distance=1pt,thick,{latex[scale=2.0]}-] (12.8,5.9)--(12.8,4.7)--(0.6,4.7)--(0.6,-6.7) --(18.3,-6.7);
\draw[line width=0.05](2.4,-3.9) --(2.4,-2.7);
\draw[line width=0.05](5.2,-3.9) --(5.2,-2.7);
\draw[line width=0.05](12.8,-3.9) --(12.8,-2.7);
\draw[line width=0.1,double distance=1pt,thick,-{latex[scale=2.0]}](2.4,-6.7) --(2.4,-5.5);
\draw[line width=0.1,double distance=1pt,thick,-{latex[scale=2.0]}](5.2,-6.7) --(5.2,-5.5);
\draw[line width=0.1,double distance=1pt,thick,-{latex[scale=2.0]}](12.8,-6.7) --(12.8,-5.5);
\draw[line width=0.05](15.6,-2.3) rectangle (14.8,-3.1);
\draw[line width=0.05,->](15.2,-3.1) --(15.2,-6.6);
\draw[line width=0.05](2.4,-2.7) --
(14.8,-2.7);
\draw[line width=0.05](3.2,0.9) rectangle (1.6,-0.7);
\draw[line width=0.05](6.0,0.9) rectangle (4.4,-0.7);
\draw[line width=0.05](13.6,0.9) rectangle (12.0,-0.7);
\draw[line width=0.05,->](15.6,2.1) --(18.7,2.1);
\draw[line width=0.05](2.4,0.9) --(2.4,2.1);
\draw[line width=0.05](5.2,0.9) --(5.2,2.1);
\draw[line width=0.05](12.8,0.9) --(12.8,2.1);
\draw[line width=0.05](15.6,2.5) rectangle (14.8,1.7);
\draw[line width=0.05,->](16.0,1.3) --(16.0,-6.6);
\draw[line width=0.05](2.4,2.1)
--(14.8,2.1);
\draw[line width=0.1,double distance=1pt,thick](18.8,-8.7) --(18.8,9.5);
\draw[line width=0.05,->](15.6,-2.7) --(18.7,-2.7);
\draw[line width=0.05](15.2,1.7) --(15.2,1.3) -- (16.0,1.3);
\draw[line width=0.1,double distance=1pt,thick,-{latex[scale=2.0]}](0.6,-1.9) --(12.8,-1.9)--(12.8,-0.7);
\draw[line width=0.1,double distance=1pt,thick,-{latex[scale=2.0]}](2.4,-1.9) --(2.4,-0.7);
\draw[line width=0.1,double distance=1pt,thick,-{latex[scale=2.0]}](5.2,-1.9) --(5.2,-0.7);
\draw[line width=0.05](2.4,7.5) --(2.4,8.7);
\draw[line width=0.05](2.4,8.7) 
--(14.8,8.7);
\draw[line width=0.05](3.2,7.5) rectangle (1.6,5.9);
\draw[line width=0.05](6.0,7.5) rectangle (4.4,5.9);
\draw[line width=0.05](5.2,7.5) --(5.2,8.7);
\draw[line width=0.1,double distance=1pt,thick,-{latex[scale=2.0]}](2.4,4.7) --(2.4,5.9);
\draw[line width=0.1,double distance=1pt,thick,-{latex[scale=2.0]}](5.2,4.7) --(5.2,5.9);
\draw[line width=0.05](13.6,7.5) rectangle (12.0,5.9);
\draw[line width=0.05](12.8,7.5) --(12.8,8.7);
\draw[line width=0.05](15.6,9.1) rectangle (14.8,8.3);
\draw[line width=0.05,->](17.0,7.9) --(17.,-6.6);
\draw[line width=0.05](15.2,8.3) --(15.2,7.9) -- (17.0,7.9);

\draw[line width=0.05,->](15.6,8.7) --(18.7,8.7);
\draw[line width=0.1,double distance=1pt,thick,-{latex[scale=2.0]}](18.8,9.5) --(20.6,9.5);
\draw[line width=0.05](15.6,-8.3) rectangle (14.8,-9.1);
\draw[line width=0.05,->](15.2,-8.3) --(15.2,-6.8);
\draw[line width=0.05,->](15.6,-8.7) --(18.7,-8.7);
\node[scale=2] at (-.7, -4.7){$\rho_1$};
\node at (2.4,-4.7){$B_1^1$};
\node at (5.2,-4.7){$B_2^1$};
\node at (12.8,-4.7){$B_{K_1}^1$};
\node[scale=2] at (8.5,-4.7){$\cdots$};
\node[scale=2] at (-.7, 0.2){$\rho_2$};
\node at (2.4,0.1){$B_1^2$};
\node at (5.2,0.1){$B_2^2$};
\node at (12.8,0.1){$B_{K_2}^2$};
\node[scale=2] at (8.5,0.1){$\cdots$};
\node[scale=2] at (-.7, 6.7){$\rho_L$};
\node at (2.4,6.7){$B_1^L$};
\node at (5.2,6.7){$B_2^L$};
\node at (12.8,6.7){$B_{K_L}^L$};
\node[scale=2] at (8.5,6.7){$\cdots$};

\node[scale=1.5] at (8.5,3.5){$\vdots$};

\node at (15.2,-9.5){$1$};
\node at (15.2,-1.9){$2$};
\node at (15.2,2.9){$3$};
\node at (15.2,9.5){$L$};

\draw[line width=0.1,double distance=1pt,thick] (19.3,-6.7) arc (0:180:0.5);
\draw[line width=0.1,double distance=1pt,thick,>-] (21.5,-6.7) -- (19.3,-6.7);

\node[style={scale=1.5}] at (15.25,10.6) {STATE};
\end{tikzpicture}
    \caption{Construction of a macro-cell. Single line represent wires transporting one bit and double lines represent multi-bit wires (representing a state).} \label{fig: compute}
  \end{subfigure}
  \caption{Recognizing one neighborhood (left), and wiring these \emph{neighborhood matchers} into a \emph{macro-cell} which computes the local function of $F$ (right).}
\end{figure}
Given $l$ and $k$, the logic gate diagram of Figure~\ref{fig: checker}, the \emph{Neighborhood Matcher}, called $B_k^l$, outputs $0$ if and only if in the input in the wires is exactly $w_k^l$ or the output cable was already in state $0$. The $i$-th wire joins the $i$-th letter in $w$ with either the $\exists$ gate on the left if the $i$-th letter of $w_k^l$ is $1$ or the $\forall$ gate on the right if the $i$-th letter of $w_k^l$ is $0$. Gate $\exists$ triggers a $0$ on wire $x$ if at least on of its incoming wire is $0$, while gate $\forall$ triggers a $0$ on wire $y$ if all incoming wires are $0$. Note that both behaviors are compatible with the freezing conditions since the set of wires in state $0$ can only grow during evolution. The gate $\alpha$ at the top triggers a $0$ on the output wire if wire $y$ is in state $0$ and wire $x$ is in state $1$ (see Figure~\ref{fig: checker}). It is also a freezing gate, meaning that once it has triggered a $0$ it will never change its state again, even if the wire $x$ turns to state $0$. Moreover this gate also turns into ``trigger'' state as soon as the output wire is $0$.

\myparagraph{Local Function Computation}
Now we can compute the local function of $F$ through a \emph{macro-cell} $\mathbb{C}_F$, receiving the states $x = (x_n)_{n \in N}$ of the neighborhood as input, and yielding the next state $f(x)$ as output.
For this we will divide the space into rows $\rho_l$ for $l \in Q_u$, and some number of columns. Intuitively, the role of row $\rho_l$ is to maintain the information ``the current state of the macro-cell is less than $l$''.
For a given $l \in Q_u$, $\rho_l$ contains all block $B_k^l$ for ${k\in\{1,\ldots,K_l\}}$. The inputs are distributed to each block, and the outputs of all blocks in $\rho_l$ are connected together by a broadcast wire. Thus, the final output in $\rho_l$ is $0$ as soon as one block $K_k^l$ triggers, \textit{i.e.} as soon as $f(x) \leq l$, see Figure~\ref{fig: compute}. Notice that once a neighborhood matcher $B^l_k$ in row $l$ has output $0$, the output of the macro-cell it belongs to must be less than $l$ for ever: indeed, at the time when the $B^l_k$ was triggered to output $0$ the output value of the Macro-Cell must be less than $l$ by definition of $B^l_k$, after that time the output is always less than $l$ thanks to the freezing condition on the CA being simulated. Concatenating rows in the right order, we obtain as output of the gate the correct state codification $f(x)$ for any state of the neighborhood $x$ received as input.

\myparagraph{Information exchanging}
Given these basic blocks, one needs to embed one \emph{macro-cell} per simulated cell on the simulator CA, and wire the inputs and outputs of neighboring macro-cells, as in Figure~\ref{fig: big block}. 
The wiring between \emph{macro-cells} depends on the neighborhood of the simulated CA. In order to clarify the presentation we will always assume that the simulated CA has a von Neumann neighborhood which is enough to achieve universality thanks to the following lemma.

\begin{lemma}
  \label{lem:reductionvonneumann}
  For any dimension $d$ and any ${F\in\FCA^d}$ there is ${G\in\FCA^d}$ with von Neumann neighborhood that simulates $F$.
\end{lemma}

The von Neumann wiring between \emph{macro-cells} in dimension $2$ is shown on Figure~\ref{fig: big block}. It is straightforward to generalize it to any dimension. Technically, thanks to Lemma~\ref{lem:reductionvonneumann}, all the encoding map $\phi$ we use later have a von Neumann neighborhood context ($C$ in Definition~\ref{def:contextsensitivesimu}).

\begin{figure}[h]
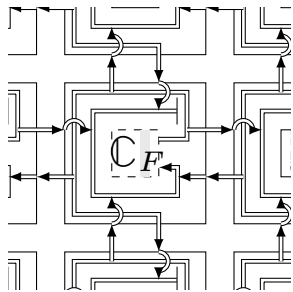

\centering
\include{./img/macrocell2}
\caption{The dashed block in the middle is a macro cell, as in figure \ref{fig: compute}.
The fat arrows exchange the state of each macro-cell with its neighbors. 
}
\label{fig: big block}
\end{figure}

\myparagraph{Context-sensitive encoding}
Given a configuration $c$ of the simulated CA, the encoding is defined as follows. All wires of the construction are in a coherent state (same state along the wire). Each macro-cell holds a state of the configuration $c$ represented in unary by the rows $\rho_l$ described above. Wires incoming from neighboring macro-cells hold the information about neighboring states (hence the context-free encoding) which is transmitted to each block $B_k^l$. Inside each of this blocks the inputs arrive at gates ``$\exists$" and ``$\forall$" and these gates have eventually triggered a $0$ on wires $x$ and $y$. However, gate ``$\alpha$" has not yet triggered to preserve the property that the main wire of row $\rho_l$ is coherent and represents the information on the \emph{current} state of the macro-cell. Starting from that well encoded configuration, a simulation cycle begins by the possible triggering of ``$\alpha$'' gates. After some time a well encoded configuration is reached again because changes coming from triggerings of $\alpha$ gates are broadcasted on each row, and, in each block $B_k^l$, the content up to the $\alpha$ gate is determined by the inputs.

We can now state three variants of the construction which differ essentially in the way crossing of information is implemented.

\begin{theorem}
  \label{thm:2change-sim} $\exists U \in \FCA_{\VN_2}$ with $5$ states which is $2$-change and $\FCA^2$-universal.
\end{theorem}

\begin{theorem}
  \label{thm:3d-sim} $\exists U \in \FCA_{\VN_3}$ with $2$ states which is $1$-change and $\FCA^3$-universal.
\end{theorem}

\begin{theorem}
  \label{thm:Moore-sim} $\exists U \in \FCA_{\MN_2}$ with $4$ states which is $1$-change and $\FCA^2$-universal.
\end{theorem}


\section{Obstacles to $\FCA$-Universality}
\label{sec:obstacles}
\myparagraph{The one-dimensional case}
\label{sec:1d}
Although one-dimensional freezing CA can be computationally universal \cite{GolOlThey15}, they cannot be $\FCA^1$-universal. This is a major difference with CA in general. The intuition behind this limitation is the following: in any given 1D freezing CA, there is a bound on the number of times a zone of consecutive cells can be crossed by a signal; and above this bound, the zone becomes a blocking word preventing any information flow between left and right halves around it.

\begin{theorem}[Dimension 1]\label{teo:1Dsim}
    There is no $F \in \FCA^1$ which is $\FCA_{\VN_1}$-universal, even with context-sensitive simulation.
\end{theorem}

\myparagraph{2D von Neumann 1-change FCA: information crossing}
\label{sec:2d_1change_VN}
We will show that there is no freezing universal FCA which is $1$-change and has the von Neumann neighborhood. This result is to be contrasted with the case of self-assembly tiling where an intrinsically universal system exists~\cite{DotyLPSSW12} (although with an unavoidably more technical definition of simulation). The intuition is that the propagation of state changes in such FCA produces $4$-connected paths that can not be crossed in the future of the evolution because only $1$ state change per cell is possible. As shown in the construction of Theorem~\ref{thm:2change-sim}, two changes per cell are enough to get rid of this limitation, even with the von Neumann neighborhood. 

We will show that no $1$-change von Neumann FCA can simulate the following 2-change FCA $(\mathbb{Z}^2,Q_F,\LN,f)$,
with ${Q_F=\{0,\leftarrow,\downarrow,\ldarrow\}}$ and where $f$ is defined by
\[\begin{matrix}
  f(0, \downarrow , 0) = \downarrow &&
  f(0,0,\leftarrow) = \leftarrow &&
  f(0,\ast,\ldarrow) = \leftarrow&&
  f(\downarrow,\ast,\leftarrow) = \ldarrow
\end{matrix}\]
and ${f(a,\ast,\ast) = a}$ else,
where $\ast$ stands for any state and the arguments of $f$ correspond to neighborhood LN in the following order: center, north, east.

\begin{theorem}
  \label{teo: nfu1D}
  There is no automaton in $\FCA_{\VN_2}$ which is 1-change and able to simulate $F$.
  Therefore there is no automaton in $\FCA_{\VN_2}$ which is 1-change and $\FCA_{\VN_2}$-universal.
\end{theorem}

\myparagraph{Monotone FCA: synchronous vs. asynchronous information}
\label{sec:monotone}
As for classical real function, we can consider the property of \emph{monotonicity} in CA:
given two configurations, one smaller that the other, their images by the CA compare in the same order.
We are particularly interested in the case where the order on configurations is given by the order on states of a freezing CA.
Several examples of such monotone FCAs were studied in literature. In particular, a simple model called \emph{bootstrap percolation} was proposed by Chalupa in  1979 \cite{Chalupa1979} to understand the properties in some magnetic materials. This model and several variants were studied from the point of view of percolation theory \cite{Gravner98,bollobas2015}, but also from the point of view of complexity of prediction \cite{Goles:2013:CBP:2535047.2535139}.

\begin{definition}
  A $\prec$-freezing CA $F$ of dimension $d$ with alphabet $Q$ is monotone if it satisfies:
    ${\forall c,c' \in Q^{\mathbb{Z}^d}: c \prec c' \Rightarrow F(c) \prec F(c')}$,
  where $\prec$ is naturally extended to configurations by
    ${c \prec c' \Leftrightarrow \forall z\in \mathbb{Z}^d: c_z \prec c'_z}$. 
\end{definition}

The intuitive limitation of monotone freezing CAs is that they must always produce a smaller state when two signals arrive simultaneously at some cell compared to when one of the two signals arrives before the other. We now exhibit a freezing non monotone CA $F$ that does precisely the opposite (non-simultaneous arrival produces a smaller state).  Next theorem shows that $F$ cannot be simulated by any freezing monotone CA.
$F$ is defined by $(\mathbb{Z}^2,\{0,s_1,s_2,w,\triangle\},LN,f)$ with $f$ given by:
\begin{multicols}{3}
\begin{itemize}
  \item $f\left(\LNpic{w}{s_i}{0}\right) = s_i$
  \item $f\left(\LNpic{w}{0}{s_i}\right) = s_i$
  \item $f\left(\LNpic{s_1}{s_2}{0}\right) = s_2$
  \item $f\left(\LNpic{s_1}{s_2}{w}\right) = s_2$
  \item $f\left(\LNpic{\triangle}{s_1}{s_1}\right) = s_1$
  \item $f\left(\LNpic{\triangle}{s_1}{w}\right) = s_2$.
\end{itemize}
\end{multicols}
and unspecified transitions let the state unchanged.
$F$ is $\prec$-freezing for the following order on states: $s_2\prec s_1\prec 0,\Delta,w$. 
Essentially $F$ is a FCA sending the signals $s_1$ and $s_2$ towards south or west along the wires materialized by state $w$. $s_2$ can also move on wires made of $s_1$. 
$\Delta$ plays the role of a non-monotone local gate: when two signals arrive simultaneously, a $s_1$-signal is sent to the south, but when only one signal arrives, a $s_2$-signal is sent to the south. $F$ cannot be simulated by any monotone FCA, hence no monotone FCA can be $\FCA$-universal.

\begin{theorem}
\label{teo: nfum}
  For any ${d\geq 1}$, there is no freezing monotone CA of dimension $d$ which is $\FCA_{\VN_d}$-universal.
\end{theorem}

\section{Perspectives}
Here are some questions or research directions that we think are worth being considered:
\begin{itemize}
\item do CA with a bounded number of change per cell also exhibit
  intrinsic universality? 
\item what about intrinsic universality for
  non-deterministic or asynchronous freezing CA as a generalization of
  aTAM models? 
\item what are the limit sets of $\FCA$? 
\item what can be said
  about $\FCA$ from an ergodic theory point of view (this includes
  questions from percolation theory)?
\end{itemize}

\bibliographystyle{plainurl}
\bibliography{UFCA}

\newpage
\section{Appendices}


\subsection{Proof of Theorem~\ref{teo: context sim}}

\begin{proof}
  Let $d\geq 1$ be any fixed dimension. By contradiction suppose that such an universal $F_u$ with alphabet $Q_u$ exists and consider for any ${n>0}$ the CA $F_n$ with states ${Q_n=\{-1,\ldots,-n\}}$ and von Neumann neighborhood with the following rule: a cell in state $q$ changes to state $r$ if ${r<q}$ and all its neighbors are in state $r$, otherwise it stays in state $q$. $F_n$ is $<$-freezing. By hypothesis $F_u$ must simulate each $F_n$ because they are all freezing CA by definition. For each $n$ let $B_n$ be the block size in the injection ${\phi_n: Q_n \rightarrow Q_u^{B_n}}$ given by simulation of $F_n$ by $F_u$ (it is a context free simulation so the context $C$ is a singleton). Since $Q_n$ is unbounded then $B_n$ is unbounded, so we can choose $n$ such that $B_n$ has at least one side which is at least two times the radius $r_u$ of the neighborhood of $F_u$. Without loss of generality we suppose that the left to right side of $B_n$ is long. Consider the configuration $x$ of $F_u$ made by a block $\phi_n(-1)$ at position $\vec{0}\in\Z^d$ surrounded by blocks $\phi_n(-2)$. Since $F_u$ on $x$ simulates $F_n$ on the configuration $x'$ made of a $-1$ surrounded by $-2$, the block at position $\vec{0}$ in $x$ must become $\phi_n(-2)$ after some time and in particular it must change: let $t_0$ be the first time such that ${F_u^{t_0}(x)}$ does not contain the block $\phi_n(-1)$ at position $\vec{0}$, and consider any position ${\vec{i}\in B_n}$ such that ${F_u^{t_0}(x)_{\vec i}\not=x_{\vec i}}$. Note that for any ${\vec k\not\in B_n}$ and any time $t$ we have ${F_u^{t}(x)_{\vec k}=x_{\vec k}}$ because cells in state $-2$ don't change during the evolution of $x'$ under $F_n$ so the corresponding blocks in the evolution of $x$ under $F_u$ don't change either: indeed, if such a block becomes different from $\phi_n(-2)$ at some time it will never become again $\phi_n(-2)$ (by the freezing condition on $F_u$ and by injectivity of $\phi_n$) thus contradicting the simulation of $F_n$ by $F_u$ through the coding $\phi_n$. Therefore it holds that ${F_u^{t_0-1}(x)=x}$ and necessarily ${t_0=1}$ so that ${F_u(x)_{\vec i}\not=x_{\vec i}}$. However, since $F_u(x)_{\vec i}$ depends only on the ${x_{\vec i+\vec z}}$ for ${\|\vec z\|_\infty\leq r_u}$ and since ${m_n\geq 2r_u}$, it is always possible to construct a pair of configurations $y$ of $F_u$ and $y'$ of $F_n$  satisfying the following conditions:
  \begin{enumerate}
  \item $y= \overline{\phi}(y')$ (\textit{i.e.} $y$ is a valid encoding of $y'$);
  \item $y'_{\vec 0}=-1$ and $F_u(x)_{\vec i}=F_u(y)_{\vec i}$;
  \item any position in $y'$ is in state $-1$ or $-2$ and has both state $-1$ and state $-2$ in its von Neumann neighborhood.
  \end{enumerate}
  Concretely, using symmetries we can suppose without loss of generality that $\vec i$ belong to the left part of the block it belongs to. Then one can choose ${y'_{\vec j}=x'_{\vec j}}$ for ${\vec j\in\{(-1,0),(-1,-1),(-1,1),(0,0),(0,-1),(0,1)\}}$ and complete it in a greedy way to satisfy condition $3$. Such a choice guaranties that ${F_u(x)_{\vec i}=F_u(y)_{\vec i}}$ because for any $\vec z$ with ${\|\vec z\|_\infty\leq r_u}$ we have ${x_{\vec i+\vec z} = y_{\vec i+\vec z}}$ (by the assumption that $\vec i$ belongs to the left part of its block and the fact the $B_n$ is long enough from left to right). $y'$ is a fixed point of $F_n$ (by condition 3) so $y$ must be a fixed point of $F_u$ (by the freezing condition on $F_u$ and the injectivity of $\phi_n$): this contradicts condition 2 which implies ${F_u(y)_{\vec i}\not=y_{\vec i}}$.
\end{proof}

\subsection{Proof of Lemma~\ref{lem:reductionvonneumann}}

\begin{proof}
  Consider a FCA ${F}$ with state set $Q$, freezing order $\prec$, neighborhood ${N=\{\vec n_1,\ldots,\vec n_k\}\subseteq\Z^d}$ and local transition map ${\delta : Q^N\rightarrow Q}$. For a suitable choice of $m$ and for each $i$ (${1\leq i\leq k}$) consider a von Neumann connected path ${P_i=(\vec p_{i,1},\ldots,\vec p_{i,m})}$ of length $m$ linking $\vec n_i$ to $\vec 0$: $\vec p_{i,1}=n_i$ and $\vec p_{i,m}=\vec 0$ and ${\vec\Delta_{i,j}=\vec p_{i,j}-\vec p_{i,j+1}\in\VN_d}$ for ${1\leq j <m}$. Now define a FCA $G$ with von Neumann neighborhood and state set made of ${mk+1}$ copies of $Q$, denoted by projections $\pi_0,\pi_{1,1}, \ldots, \pi_{k,m}$ from ${Q^{mk+1}\rightarrow Q}$, and with the following behavior at each step:
  \begin{itemize}
  \item ${\pi_0(G(c)_{\vec z}) = \delta\bigl(\pi_{1,m}(c_{\vec z}),\ldots,\pi_{k,m}(c_{\vec z})\bigr)}$ if ${\delta\bigl(\pi_{1,m}(c_{\vec z}),\ldots,\pi_{k,m}(c_{\vec z})\bigr)\prec \pi_0(c_{\vec z})}$ and $\pi_0(c_{\vec z})$ otherwise,
  \item ${\pi_{i,j+1}(G(c)_{\vec z}) = \pi_{i,j}(c_{\vec z+\vec \Delta_{i,j}})}$ if ${\pi_{i,j}(c_{\vec z+\vec\Delta_{i,j}})\prec \pi_{i,j+1}(c_{\vec z})}$ and ${\pi_{i,j+1}(c_{\vec z})}$ otherwise, for ${1\leq j<m}$ and ${1\leq i\leq k}$,
  \item ${\pi_{i,1}(G(c)_{\vec z})= \pi_0(c_{\vec z})}$ if ${\pi_0(c_{\vec z})\prec \pi_{i,1}(c_{\vec z})}$ and $\pi_{i,1}(c_{\vec z})$ otherwise, for ${1\leq i\leq k}$.
  \end{itemize}
  Intuitively, $G$ realizes in parallel the propagation of neighboring $Q$-states along paths of the form ${z+P_i}$ from any cell ${\vec z}$ and the application of the local transition $\delta$ in each cell using the $Q$-components corresponding to the end of each propagation path $P_i$.  Let's show that $G$ is a freezing CA that simulates $F$. First, it is clear that $G$ is freezing because in any transition, any component of the state can only decrease according to $\prec$. Now consider the encoding map ${\phi: Q^{\Z^d}\rightarrow \left(Q^{km+1}\right)^{\Z^d}}$ defined by:
  \begin{itemize}
  \item ${\pi_0(\phi(c)_{\vec z}) = \pi_{i,1}(\phi(c)_{\vec z}) = c_{\vec z}}$ for ${1\leq i\leq k}$,
  \item ${\pi_{i,j+1}(\phi(c)_{\vec z}) = \pi_{i,j}(\phi(c)_{\vec z+\vec \Delta_{i,j}})}$ for ${1\leq j<m}$ and ${1\leq i\leq k}$.
  \end{itemize}
  A configuration $\phi(c)$ is such that the $Q$-value is constant along $P_i$ paths so only the $\pi_0$ components can change when applying $G$ and we necessarily have ${\pi_0(G(\phi(c))_{\vec z}) = F(c)_{\vec z}}$ for ${1\leq i\leq k}$, because $F$ is freezing for the order $\prec$ and ${\pi_{i,m}(\phi(c)_{\vec z}) = c_{\vec z+ \vec n_i}}$ by the second item above and the definition of paths $P_i$. Then, in ${G(\phi(c))}$, only the $\pi_{i,1}$ components can change and it holds that 
\[\pi_{i,1}(G^2(\phi(c))_{\vec z}) = \pi_{0}(G(\phi(c))_{\vec z}) = F(c)_{\vec z}\]
for ${1\leq i\leq k}$. Similarly it is straightforward to check that after ${m+1}$ steps the two item of the definition of $\phi$ are again verified and we have: ${G^{m+1}(\phi(c)) = \phi(F(c))}$. This shows that $G$ simulates $F$ and the lemma follows.

\end{proof}

\subsection{Proof sketch of Theorem~\ref{thm:2change-sim}}
\label{sec:2d_2change_VN}
\newcommand{\rsquare}{{\color[RGB]{200,0,0} \blacksquare}}
\definecolor{reddish}{RGB}{200,0,0}
\newcommand\sgatewire[1]{
  \begin{tikzpicture}[baseline=-0.65ex]\draw (0,0) node[draw=black] {\tiny #1};
    \end{tikzpicture}
  }
\newcommand\sgatefire[1]{
  \begin{tikzpicture}[baseline=-0.65ex]
    \draw (0,0) node[fill=reddish] {\tiny #1};
  \end{tikzpicture}}

\newcommand\sbg{\square}
\newcommand\swire{\blacksquare}
\newcommand\sfire{\rsquare}
\newcommand\shori{\sgatewire{$\leftrightarrow$}}
\newcommand\svert{\sgatewire{$\updownarrow$}}
\newcommand\salphawire{\sgatewire{$\alpha$}}
\newcommand\salphafire{\sgatefire{$\alpha$}}
\newcommand\sexistswire{\sgatewire{$\exists$}}
\newcommand\sexistsfire{\sgatefire{$\exists$}}
\newcommand\sforallwire{\sgatewire{$\forall$}}
\newcommand\sforallfire{\sgatefire{$\forall$}}

We can make a direct implementation of the abstract construction as a universal FCA is $U=\{\Z^2,\{\sbg,\swire,\sfire,\shori,\svert, \salphawire,\salphafire,\sexistswire,\sexistsfire,\sforallwire,\sforallfire
\},VN,f_u\}$. This is a 2-change FCA with freezing order: \[\sbg,\swire,\salphawire,\sexistswire,\sforallwire\geq \shori,\svert\geq \sfire, \salphafire,\sexistsfire,\sforallfire.\]
It implements all elements of Section~\ref{sec:abstract-sim} and the states have the following meaning:
\begin{itemize}
\item $\sbg$ is the quiescent background,
\item $\swire$ is a wire waiting for a signal,
\item $\sfire$ is a signal,
\item $\shori$ (resp. $\svert$) an intermediate states to manage a crossing when a first signal already passed horizontally (resp. vertically),
\item $\salphawire$ (resp. $\sexistswire$ and $\sforallwire$) is the $\alpha$ (resp. $\exists$ and $\forall$) gate waiting the conditions to trigger,
\item $\salphafire$ (resp. $\sexistsfire$ and $\sforallfire$) is the $\alpha$ (resp. $\exists$ and $\forall$) gate once it has triggered.
\end{itemize}

All wires described by the abstract construction are made by drawing
trees of degree at most $3$ of VN-connected cells in state $\swire$.
The case of $\swire$ with $4$ neighbors in state $\swire$ is reserved
to manage crossings. Also gates $\exists$ and $\forall$ have unbounded
fan-in in the abstract construction. Here we simulate unbounded fan-in
by fan-in $2$ (which is possible because the semantics of these gates
is associative). More precisely, all gates
($\salphawire$,$\sexistswire$ and $\sforallwire$) receive there first
(resp. second) input from their left (resp. right) neighbor and send
their output to the top.

The local rule is given by the following set of transitions, any cell which is in a local context not appearing in this list stays unchanged:
\newcommand\transi[6]{
  \begin{matrix}
    &#2\\
    #5&#1&#3\\
    &#4
  \end{matrix}\mapsto #6
}
\begin{description}
\item[Normal wires:] if ${\sbg\in\{n,e,s,w\}}$
  \begin{itemize}
  \item $\transi{\swire}{n}{e}{s}{w}{\sfire}$ if
    ${\sfire\in\{n,e,s,w\}}$
  \item $\transi{\swire}{n}{e}{s}{w}{\sfire}$ if
    ${\shori\in\{e,w\}}$ or ${\svert\in\{n,s\}}$
  \end{itemize}
\item[Crossings:] if ${\{n,e,s,w\}\subseteq\{\swire,\sfire\}}$
  \begin{itemize}
  \item $\transi{\swire}{n}{e}{s}{w}{\shori}$ if ${\sfire\in\{e,w\}}$
    and ${\sfire\not\in\{n,s\}}$
  \item $\transi{\swire}{n}{e}{s}{w}{\svert}$ if ${\sfire\not\in\{e,w\}}$ and ${\sfire\in\{n,s\}}$
  \item $\transi{\swire}{n}{e}{s}{w}{\sfire}$ if ${\sfire\in\{e,w\}}$ and ${\sfire\in\{n,s\}}$
  \item $\transi{\shori}{n}{e}{s}{w}{\sfire}$ if ${\sfire\in\{n,s\}}$
  \item $\transi{\svert}{n}{e}{s}{w}{\sfire}$ if ${\sfire\in\{e,w\}}$
  \end{itemize}
\item[Gates triggering:] if ${\{e,w\}\subseteq\{\swire,\sfire\}}$
  \begin{itemize}
  \item $\transi{\salphawire}{n}{\swire}{s}{\sfire}{\salphafire}$,
  \item $\transi{\sexistswire}{n}{e}{s}{w}{\sexistsfire}$ if ${\sfire\in\{e,w\}}$,
  \item $\transi{\sforallwire}{n}{\sfire}{s}{\sfire}{\sforallfire}$
  \end{itemize}
\item[Gates output:]\ \\
  \begin{itemize}
  \item $\transi{\swire}{n}{e}{s}{w}{\sfire}$ if ${s\in\{\salphafire,\sexistsfire,\sforallfire\}}$
  \end{itemize}
\end{description}

It appears that the south state in the gate triggering transitions above is not used. Moreover no transition involves the background state $\sbg$ and the behavior of $\svert$ is similar to that of gate output transitions. This allows us to reduce the state set to ${\{\sbg,\swire,\sfire,\shori,\svert\}}$ and to code all triggered gates by $\svert$ and (untriggered) gates $\salphawire$ (resp. $\sexistswire$ and $\sforallwire$) by a $\sbg$ state having at south a $\sbg$ (resp. $\shori$ and $\svert$). More precisely, we keep all transitions for normal wires and crossings and add the following ones which replace the gates triggering:
if ${\{n,e,w\}\subseteq\{\swire,\sfire\}}$
\begin{itemize}
\item $\transi{\sbg}{n}{\swire}{\sbg}{\sfire}{\svert}$,
\item $\transi{\sbg}{n}{e}{\shori}{w}{\svert}$ if ${\sfire\in\{e,w\}}$,
\item $\transi{\sbg}{n}{\sfire}{\svert}{\sfire}{\svert}$.
\end{itemize}

The gates output transitions are already realized by the behavior of
$\svert$ on normal wires. Finally, it is important to note that the crossing
transition that transforms a $\svert$ into $\sfire$ will not interfere
here because it applies only when all states surrounding $\svert$ belongs to ${\{\swire,\sfire\}}$ and in the $3$ transitions above to simulate gates, the $\svert$ state generated will have a state among ${\{\sbg,\shori,\svert\}}$ as south neighbor. We conclude that $5$ states are enough to achieve universality with von Neumann neighborhood in 2D.

\subsection{Proof sketch of Theorem~\ref{thm:3d-sim}}
\label{sec:3d}
The abstract construction works exactly the same way in 3D, the only difference being that there are more neighbors in the 3D version of von Neumann neighborhood. Moreover, the 2D CA constructed in Section~\ref{sec:2d_2change_VN} can be used almost as is to obtain a 3D FCA-universal example. Indeed, all the logic circuitry of macro cell can be done in a planar way and the third dimension matters only in the wiring between 3D macro-cells.

\begin{figure}[H]
    \centering


\tdplotsetmaincoords{60}{60}
\begin{tikzpicture}[scale=0.3]

\pgfmathsetmacro{\cubex}{1}
\pgfmathsetmacro{\cubey}{1}
\pgfmathsetmacro{\cubez}{1}

\cubo{0}{ 0}{0}
\cubo{1}{ 0}{0}

\cubo{4}{-3}{0}
\cubo{4}{-2}{0}
\cubo{4}{-1}{0}
\cubo{4}{ 0}{0}
\cubo{4}{ 1}{0}

\cubo{2}{ 0}{0}
\cubo{2}{ 0}{1}
\cubo{2}{ 0}{2}
\cubo{3}{ 0}{2}
\cubo{4}{ 0}{2}
\cubo{5}{ 0}{2}

\cubo{6}{ 0}{0}
\cubo{6}{ 0}{1}
\cubo{6}{ 0}{2}

\cubo{7}{ 0}{0}
\cubo{8}{ 0}{0}

\end{tikzpicture}
    \caption{Crossing signal in 3-D.}\label{fig: 3d}
\end{figure}
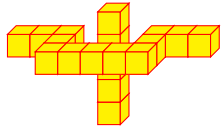%
 To do so, it is sufficient to add the 3-dimensional equivalent of the normal wires transition of Section~\ref{sec:2d_2change_VN}. 
Moreover, we can build a 3D FCA-universal CA which is only 1-change. This is possible by substituting all the crossings mechanics used in Section~\ref{sec:2d_2change_VN} for the crossing given in the Figure~\ref{fig: 3d}. Note that crossing transitions are the only place where the intermediate states $\shori $ and $\svert$ can disappear. Therefore, when removing those transitions, we get a 1-change FCA with freezing order: \[\sbg,\swire\geq \shori,\svert,\sfire.\]
At this point states $\shori$ and $\svert$ are totally unrelated to crossings and are just used to code the behavior in the gates triggering transitions and the propagation of a $\sfire$ at triggered gate outputs. However a cell in 3D has 6 von Neumann neighbors, therefore there is room to code all the different behaviors using less states. 

In fact 2 states are enough and FCA-universality is achieved by the 3D von Neumann FCA on ${\{0,1\}}$ given by:
\begin{itemize}
\item a $1$ surrounded by exactly two $0$s becomes $0$;
\item $0$s stay unchanged.
\end{itemize}
The 2D version of this FCA was studied in \cite{GolesMMO17} were it was shown that it can implement all necessary synchronous logical gates. By using such planar constructions in the 3D version and using the third dimension to implement asynchronous crossings as above, we can realize the abstract construction of Section~\ref{sec:abstract-sim}.

\subsection{Proof sketch of Theorem~\ref{thm:Moore-sim}}
\label{sec:2d_Moore}

We build an intrinsically universal FCA $U_M$ with Moore Neighborhood.
\[U_M=\cadef(\mathbb{Z}^2,\{\sbg , \swire,\sfire,\svert,\shori\},\MN,f_M\})\,.\]

The local function $f_M$ is an extension of the CA defined in Section~\ref{sec:2d_2change_VN}, adding rules using the Moore neighborhood to build a crossing as follows:

\begin{multicols}{3}
\begin{itemize}
  \item $f_M\left(\MNpicf{}{$\rbsquare$     }{$\swire$    }{}{$\swire$    }{$\rbsquare$     }{$\swire$    }{}{}\right) = \swire$
  \item $f_M\left(\MNpicf{}{$\rbsquare$     }{$\sfire$}{}{$\swire$    }{$\rbsquare$     }{$\swire$    }{}{}\right) = \sfire$
  \item $f_M\left(\MNpicf{}{$\rbsquare$     }{$\swire$    }{}{$\swire$    }{$\rbsquare$     }{$\sfire$}{}{}\right) = \sfire$
  \item $f_M\left(\MNpicf{}{}{$\swire$    }{}{$\swire$    }{}{}{$\swire$}{}\right) = \swire$
  \item $f_M\left(\MNpicf{}{}{$\sfire$}{}{$\swire$    }{}{}{$\swire$}{}\right) = \sfire$
  \item $f_M\left(\MNpicf{}{}{$\swire$}{}{$\swire$    }{}{}{$\sfire$}{}\right) = \sfire$\,,
\end{itemize}
\end{multicols}

and the rotations and reflexions of these transitions, where $\rbsquare$ match both $\swire$ and $\sfire$ , and 

\begin{align*}
  f_M\left(\MNpic{$\alpha$}{$a$}{$\beta$}{$e$}{$c$}{$b$}{$\gamma$}{$d$}{$\delta$}\right) = f_u\left(\VNpic{$c$}{$a$}{$b$}{$d$}{$e$}\right) \ 
  \text{if \MNpic{$\alpha$}{$a$}{$\beta$}{$e$}{$c$}{$b$}{$\gamma$}{$d$}{$\delta$} is not in the previous cases}\,
\end{align*} 
where $f_u$ is the $5$-states CA defined in Section~\ref{sec:2d_2change_VN} without the crossings transition. It is therefore a 1-change CA with freezing order ${\sbg,\swire\geq\sfire,\shori,\svert}$.

With this local function the realization of crossings is given by Figure~\ref{fig: moore}.

\begin{figure}[H]
  \begin{subfigure}{.5\textwidth}
    \centering

\begin{tikzpicture}[scale=0.3, every node/.style={scale=0.6}]

\fill[black] (2,4)  -- ++(1,0) -- ++(0,1) -- ++(-1,0) -- cycle;
\fill[black] (2,3)  -- ++(1,0) -- ++(0,1) -- ++(-1,0) -- cycle;
\fill[black] (3,2)  -- ++(1,0) -- ++(0,1) -- ++(-1,0) -- cycle;
\fill[black] (4,1)  -- ++(1,0) -- ++(0,1) -- ++(-1,0) -- cycle;
\fill[black] (5,0)  -- ++(1,0) -- ++(0,1) -- ++(-1,0) -- cycle;
\fill[black] (5,-1)  -- ++(1,0) -- ++(0,1) -- ++(-1,0) -- cycle;
\fill[black] (5,3)  -- ++(1,0) -- ++(0,1) -- ++(-1,0) -- cycle;
\fill[black] (4,2)  -- ++(1,0) -- ++(0,1) -- ++(-1,0) -- cycle;
\fill[black] (3,1)  -- ++(1,0) -- ++(0,1) -- ++(-1,0) -- cycle;
\fill[black] (2,0)  -- ++(1,0) -- ++(0,1) -- ++(-1,0) -- cycle;
\fill[black] (2,-1)  -- ++(1,0) -- ++(0,1) -- ++(-1,0) -- cycle;
\fill[black] (5,4)  -- ++(1,0) -- ++(0,1) -- ++(-1,0) -- cycle;
\draw  (1,5) rectangle (7,-1);
\draw [help lines, step=1cm] (1,-1) grid (7,5);
\node at (2.5,5.5) {In 1};
\node at (5.5,5.5) {In 2};
\node at (2.5,-2) {Out 1};
\node at (5.5,-2) {Out 2};
\end{tikzpicture}  
  \end{subfigure}
  \hfill%
  \begin{subfigure}{.5\textwidth}
    \centering

\begin{tikzpicture}[scale=0.3, every node/.style={scale=0.6}]

  \fill[reddish] (2,4)  -- ++(1,0) -- ++(0,1) -- ++(-1,0) -- cycle;
  \fill[reddish] (2,3)  -- ++(1,0) -- ++(0,1) -- ++(-1,0) -- cycle;
  \fill[reddish] (3,2)  -- ++(1,0) -- ++(0,1) -- ++(-1,0) -- cycle;
  \fill[reddish] (4,1)  -- ++(1,0) -- ++(0,1) -- ++(-1,0) -- cycle;
  \fill[reddish] (5,0)  -- ++(1,0) -- ++(0,1) -- ++(-1,0) -- cycle;
  \fill[reddish] (5,-1)  -- ++(1,0) -- ++(0,1) -- ++(-1,0) -- cycle;
  \fill[black] (5,3)  -- ++(1,0) -- ++(0,1) -- ++(-1,0) -- cycle;
  \fill[black] (4,2)  -- ++(1,0) -- ++(0,1) -- ++(-1,0) -- cycle;
  \fill[black] (3,1)  -- ++(1,0) -- ++(0,1) -- ++(-1,0) -- cycle;
  \fill[black] (2,0)  -- ++(1,0) -- ++(0,1) -- ++(-1,0) -- cycle;
  \fill[black] (2,-1)  -- ++(1,0) -- ++(0,1) -- ++(-1,0) -- cycle;
  \fill[black] (5,4)  -- ++(1,0) -- ++(0,1) -- ++(-1,0) -- cycle;
\draw  (1,5) rectangle (7,-1);
\draw [help lines, step=1cm] (1,-1) grid (7,5);
\node at (2.5,5.5) {In 1};
\node at (5.5,5.5) {In 2};
\node at (2.5,-2) {Out 1};
\node at (5.5,-2) {Out 2};
\end{tikzpicture}  
  \end{subfigure}
  \caption{Crossing signal in Moore neighborhood. 
    If the cell in In 1 (2) is in state $\sfire$, then after five iterations the cell in Out 2 (1) change to state $\sfire$ (as seen on the right).}\label{fig: moore}
\end{figure}
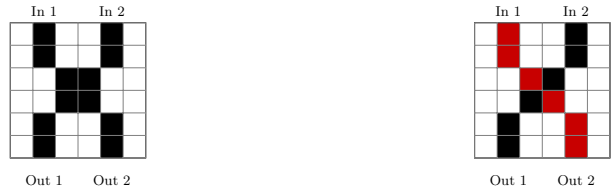%

As in previous constructions there is room for optimization of the number of states. For instance, one can remove the $\shori$ state used in $f_u$ to encode a type of gate. Instead we can use the space allowed by the Moore neighborhood to encode the gate by using different patterns of $\sbg$ and $\svert$ in the bottom row of the neighborhood. This give a FCA-universal example with only $4$ states.

A natural candidate with only 2 states in this setting is ``life without death''. It is quite possible that it is FCA-universal. However the circuitry built in~\cite{GriMoo96} to show the P-completeness of this CA cannot be used directly here, in particular it does not yield an implementation of asynchronous crossing. We leave the question of FCA-universality of ``life without death'' open.



\subsection{Proof of Theorem~\ref{teo:1Dsim}}

\begin{proof}
  Suppose that $F_u$ is a freezing 1D CA with radius $r$ and alphabet $Q_u$ that can simulate any freezing 1D CA. There is a constant ${M\leq |Q_u|^r}$ such that the global state of a group of $r$ consecutive cells of $F_u$ cannot change more than $M$ times.
  Consider now the CA $F$ with states ${Q_M=\{0,\ldots,M+1\}}$ defined by: ${F(c)_i = \min(c_i,c_{i+1})}$ for any configuration $c$ and any ${i\in\Z}$. $F$ is a freezing CA for the natural order on integers so by hypothesis $F_u$ simulates $F$ using (context-sensitive) encoding map ${\overline\phi: Q_M^\Z\rightarrow Q_u^\Z}$ defined from a local map $\phi$ as in Definition~\ref{def:contextsensitivesimu}. If a configuration $c$ of $F$ is such that 
  \[c_i =
  \begin{cases}
    q &\text{ if } -k\leq i\leq k\\
    \geq q &\text{ if }i<k\\
    <q &\text{ if }i>k\\
  \end{cases}
\]
for $k$ larger than the radius of map $\phi$, then there is $t$ such that ${\overline\phi(c)_{[0,r-1]}\not=F_u^t(\overline\phi(c))_{[0,r-1]}}$: indeed if ${F_u^t(\overline\phi(c))_{[0,r-1]}}$ stays constant with $t$, then, considering the configuration $d$ such that ${c_i=d_i}$ for ${i\leq k}$ and ${d_i=q}$ for ${i>k}$, we also have that ${F_u^t(\overline\phi(d))_{[0,r-1]}}$ is constant (because ${F^t(d)_i}$ is constant for any ${i\geq -k}$ and $k$ is larger than the radius of $\phi$). Moreover ${\overline\phi(d)_i=\overline\phi(c)_i}$ for ${i<r}$, therefore we should have ${F_u^t(\overline\phi(c))_i=F_u^t(\overline\phi(d))_i}$ for all ${i<r}$ and all $t$ which would contradict the simulation of $F$ by $F_u$ since ${F^t(c)_i}$ and ${F^t(d)_i}$ differ for some $t$ and ${i<0}$.

Consider the following configuration of $F$: 
\[d = {}^\infty(M+1)\cdot M^{k_1}(M-1)^{k_2}3^{k_3}\cdots 1^{k_{M}}0^\infty\]
where the leftmost occurrence of state $M$ is at position ${i=k_0}$ for a choice of large enough ${k_0,\ldots,k_{M}}$. By the reasoning above, we must have that ${\bigl(F_u^t(\phi(d))_{[0,r-1]}\bigr)_t}$ must change ${M+1}$ times but this contradicts the definition of $M$.
\end{proof}

\subsection{Proof of Theorem~\ref{teo: nfu1D}}

Given a FCA $F$ with von Neumann neighborhood, we call \emph{changing path} from $z$ to $z'$ between configurations $c$ and ${F^t(c)}$ a path $z_1,\ldots,z_n$ such that:
\begin{itemize}
  \item $z_1=z$,
  \item $z_n=z'$,
  \item $z_{i+1} \in z_{i}+\VN_2,\ \forall i=1,\ldots,n $,
  \item $c_{z_i}\neq F^t(c)_{z_i},\ \forall i=1,\ldots,n$.
\end{itemize}
  
We also say that a position $z$ is \emph{stable} in a configuration $c$ if ${F(c)_z = c_z}$ and \emph{unstable} if it is not stable.
 
\begin{lemma}[Changing path lemma]
  \label{lem:changepath}
  Let $z$ be a position in some configuration $c$ and let ${t\geq 1}$ be such that ${F^t(c)_z\neq c_z}$, then there exists an unstable position $z'$ in $c$ and a changing path of length at most $t$ from $z$ to $z'$ between configurations $c$ and ${F^t(c)}$.
\end{lemma}
\begin{proof}
  Consider the first time $t'$ such that ${F^{t'}(c)_z\neq c_z}$. If ${t'=1}$ we are done because in this case $z$ itself is unstable in $c$. Otherwise, $z$ is stable in $c$ and therefore one of its neighbors must have changed before time $t'$. Therefore we can apply inductively the lemma on this neighbor with a time ${t<t'\leq t}$ to get a changing path of length at most ${t-1}$ from an unstable position to this neighbor, which we complete into a changing path of length at most $t$ from an unstable position to $z$.
\end{proof}

\begin{proof}[Proof of Theorem~\ref{teo: nfu1D}]
Suppose by contradiction that such a 1-change FCA ${U=(\mathbb{Z}^2,Q,VN,f_s)}$ exists. Let ${\phi : Q_F^C\rightarrow Q^B}$ be the encoding map ensuring the simulation of $F$ by $U$ in $T$ steps:
 \[\overline{\phi}(F(c))=U^T(\overline{\phi}(c))\]
with ${\overline\phi(c)_{\vec z}}$ depending exactly on cells ${\left\lfloor\vec z/b\right\rfloor + C}$ of $c$, where ${b\in\Z^2}$ is the size-vector of $B$. The injectivity of $\overline\phi$, the simulation and Lemma~\ref{lem:changepath} ensure that there is a finite set ${E\subseteq\Z^2}$ (depending on $C$, $b$ and $T$) such that for any configuration ${c\in Q_F^{\Z^2}}$:
\begin{itemize}
\item if some position $z$ is unstable in $c$ then some position ${z'\in bz+E}$ is unstable in ${\overline\phi(c)}$;
\item if all positions in ${z+C}$ are stable in $c$ then all positions in ${bz+B}$ are stable in ${\overline\phi(c)}$.
\end{itemize}
Let us consider configuration ${c^{n}\in Q_F^{\Z^2}}$ for any ${n\geq 0}$ defined by:
\[c^{n}(z) =
\begin{cases}
  \downarrow&\text{ if }z=(0,n),\\
  \leftarrow&\text{ if }z=(n^2,0),\\
  0&\text{ else.}
\end{cases}
\]
By definition of $F$, for any $t$, ${F^t(c^{n})}$ contains exactly two unstable positions: ${(0,n-t)}$ and ${(n^2-t,0)}$ ($\downarrow$ propagates downward, $\leftarrow$ propagates to the left eventually crossing a $\downarrow$ and everything else stays unchanged). Using the observations above and Lemma~\ref{lem:changepath} we deduce that for any large enough $n$ there exist a changing path ${P_n=(z_1,\ldots,z_m)}$ of length ${\Omega(n)}$ from ${z_1\in b\cdot(0,n)+E}$ to ${z_n\in b\cdot(0,-n)+E}$ between configurations ${\overline\phi(c^{n})}$ and ${U^{2nT}(\overline{\phi}(c^{n}))}$. Moreover, choosing $n$ large enough, each position of the path $P_n$ is at distance at most $K$ from the vertical axis where $K$ is a constant given by the simulation that do not depend on $n$: indeed changes between ${c^{n}}$ and ${F^{2n}(c^{n})}$ occur on the vertical axis or at distance at least ${n^2-2n}$ from it. Then, for a suitable choice of a ${t_0\in o(n)}$ we have that:
\begin{itemize}
\item there is an unstable position ${z=(x,y)}$ in ${U^{(n^2-t_0)T}(\overline{\phi}(c^{n}))}$ with ${y>K}$ and at distance ${o(n)}$ from the center $(0,0)$, while all other unstable positions are at distance ${\Omega(n)}$ from the center;
\item there is a position ${z'=(x',y')}$ with ${y'<K}$ and at distance $o(n)$ from the center such that ${U^{(n^2+t_0)T}(\overline{\phi}(c^{n}))_{z'}\neq\overline{\phi}(c^{n})_{z'}}$.
\end{itemize}

\begin{figure}[H]
  \centering
  \begin{tikzpicture}[scale=.5]
    \draw[dotted] (-1,5) node[above] {$-K$} --(-1,-5);
    \draw[dotted] (1,5) node[above] {$K$} --(1,-5);
    \draw[dotted,->] (0,-5) -- (0,6);
    \draw[dotted,->] (-2,0) -- (3,0);
    \draw[very thick,draw=blue] plot [smooth, tension=1] coordinates {(0,5)  (-0.1,3)  (-0.5,0)  (0.1,-1) (0,-3) (0.2,-5)};
    \draw[very thick,draw=red] plot [smooth, tension=1] coordinates {(-1.5,0)  (-1,0.2)  (-0.5,-.2)  (1,.1) (1.5,-.3)};
    \draw (-0.1,3) node[left] {$P_n$};
    \draw (1.5,-.3) node[below right] {$|P'_n|\in o(n)$};
    \draw[<->] (-2,5)-- node[midway,left] {$\Omega(n)$} (-2,-5);
  \end{tikzpicture}
  \caption{Changing paths $P_n$ and $P'_n$ that must cross each other.}
  \label{fig:paths}
\end{figure}
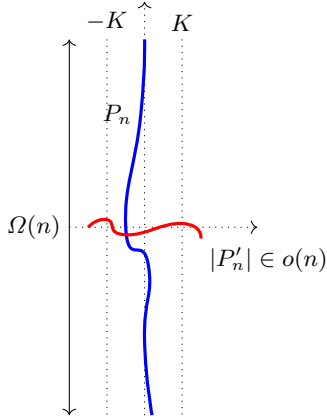

Using Lemma~\ref{lem:changepath} again, we deduce the existence of a changing path $P'_n$ of length ${o(n)}$ from $z$ to $z'$ between configurations ${U^{(n^2-t_0)T}(\overline{\phi}(c^{n}))}$ and ${U^{(n^2+t_0)T}(\overline{\phi}(c^{n}))}$. Given the respective length and endpoints of $P_n$ and $P'_n$ (see Figure~\ref{fig:paths}), they must necessarily cross each other: this is a contradiction because all positions of $P_n$ have already made a change under the action of $U$ after time ${2nT}$, so none of them can change later since $U$ is $1$-change.
\end{proof}

\subsection{Proof of Theorem~\ref{teo: nfum}}

The key to understand the limitations of monotone FCAs is to establish that some configurations is 'above' another one for the freezing order, and then use the fact that this relation is preserved under iterations. The next lemma gives a tool to obtain such relations in the context of a simulation between FCAs.

Given a CA $F$ and a finite set ${E\subseteq\Z^d}$, we say that a configuration $y$ is \emph{$E$-locally reachable} from a configuration $x$ if for any ${\vec{i}\in\Z^d}$ there are configurations $x^i$ and $y^i$ with:
\begin{itemize}
\item ${x^{\vec i}_{|\vec i+E}=x_{|\vec i+E}}$;
\item ${y^{\vec i}_{|\vec i+E}=y_{|\vec i+E}}$;
\item ${F^t(x^{\vec i})=y^{\vec i}}$ for some ${t\geq 0}$.
\end{itemize}

\begin{lemma}[Local reachability lemma]
  \label{lem:localreach}
  Let $G$ be a $\prec$-freezing CA that (context sensitively) simulates a CA $F$, both of dimension $d$. Then there exists a finite set ${E\subseteq\Z^d}$ such that, if a configuration $y$ is $E$-locally reachable by $F$ from a configuration $x$, then ${\overline\phi(y)\prec\overline\phi(x)}$ where $\overline\phi$ is the encoding map given by the simulation as in Definition~\ref{def:contextsensitivesimu}.
\end{lemma}
\begin{proof}
  Using the notations of definition~\ref{def:contextsensitivesimu}, there exists some finite set ${E\subseteq\Z^d}$ such that for any ${x,y\in Q_F^d}$ and any ${\vec i\in\Z^d}$ it holds:
  \[x|_{\vec i+E}=y|_{\vec i+E}\Rightarrow \overline\phi(x)|_{b\vec i + B} = \overline\phi(y)|_{b\vec i + B}.\]
  Suppose now that $y$ is $E$-locally reachable from $x$ by $F$. We have in particular ${\overline\phi(y^{\vec i})\prec\overline\phi(x^{\vec i})}$ because $G$ is $\prec$-freezing and the simulation ensures that ${\overline\phi(y^{\vec i})}$ is in the orbit of of ${\overline\phi(x^{\vec i})}$ under $G$ because $y^{\vec i}$ is in the orbit of $x^{\vec i}$ under $F$. From the $E$-locality condition and the remark above we deduce: 
  \[\forall j\in B:\ \overline\phi(y)_{b\vec i + j} \prec \overline\phi(x)_{b\vec i + j}\]
and since the relation holds for any ${\vec i}$ we finally have: ${\overline\phi(y)\prec\overline\phi(x)}$.
\end{proof}

If $F$ is a freezing CA and $c$ is any configuration then the limit ${\displaystyle \lim_{t\to \infty} F^t(c)}$ always exist (in the Cantor topology), is always a fix point, and will be denoted ${F^\infty(c)}$ \cite{GolOlThey15}.

\begin{proof}[Proof of Theorem~\ref{teo: nfum}]
  The case of dimension $1$ is already handled by Theorem~\ref{teo:1Dsim}. We do the proof for ${d=2}$ using $F$ defined above. It is straightforward to extend the argument to higher dimensions. Suppose by contradiction that there is a monotone $\prec$-freezing $G$ that can simulate $F$ and let $E$ be the set given by Lemma~\ref{lem:localreach} for this simulation. Let us define $x^n$, $x^\infty$ and $y^n$, $y^\infty$ as follows (see Figure~\ref{fig:confsy}):
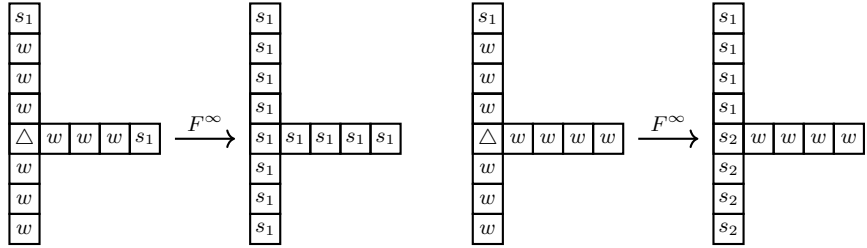
\begin{figure}[H]
  \centering
  \hfill
  \begin{subfigure}[b]{0.45\textwidth}
    \centering
\begin{tikzpicture}[thick,scale=0.4, every node/.style={scale=0.9}]

\draw  (0,0) rectangle (1,1);
\draw  (2,-4) rectangle (3,-3);
\draw  (1,-4) rectangle (2,-3);
\draw  (4,-4) rectangle (5,-3);
\draw  (3,-4) rectangle (4,-3);
\draw  (0,-1) rectangle (1,0);
\draw  (0,-3) rectangle (1,-2);
\draw  (0,-2) rectangle (1,-1);
\draw  (0,-5) rectangle (1,-2);
\draw  (0,-4) rectangle (1,-3);
\draw  (0,-7) rectangle (1,-6);
\draw  (0,-6) rectangle (1,-5);
\node at (0.5,0.5) {$s_1$};
\node at (0.5,-0.5) {$w$};
\node at (0.5,-1.5) {$w$};
\node at (0.5,-2.5) {$w$};
\node at (0.5,-3.5) {$\triangle$};
\node at (0.5,-4.5) {$w$};
\node at (0.5,-5.5) {$w$};
\node at (0.5,-6.5) {$w$};

\node at (1.5,-3.5) {$w$};
\node at (2.5,-3.5) {$w$};
\node at (3.5,-3.5) {$w$};
\node at (4.5,-3.5) {$s_1$};
\draw  (8,0) rectangle (9,1);
\draw  (10,-4) rectangle (11,-3);
\draw  (9,-4) rectangle (10,-3);
\draw  (12,-4) rectangle (13,-3);
\draw  (11,-4) rectangle (12,-3);
\draw  (8,-1) rectangle (9,0);
\draw  (8,-3) rectangle (9,-2);
\draw  (8,-2) rectangle (9,-1);
\draw  (8,-5) rectangle (9,-2);
\draw  (8,-4) rectangle (9,-3);
\draw  (8,-7) rectangle (9,-6);
\draw  (8,-6) rectangle (9,-5);
\node at (8.5,0.5) {$s_1$};
\node at (8.5,-0.5) {$s_1$};
\node at (8.5,-1.5) {$s_1$};
\node at (8.5,-2.5) {$s_1$};
\node at (8.5,-3.5) {$s_1$};
\node at (8.5,-4.5) {$s_1$};
\node at (8.5,-5.5) {$s_1$};
\node at (8.5,-6.5) {$s_1$};

\node at (9.5,-3.5) {$s_1$};
\node at (10.5,-3.5) {$s_1$};
\node at (11.5,-3.5) {$s_1$};
\node at (12.5,-3.5) {$s_1$};
\draw[->] (5.5,-3.5) -- (7.5,-3.5);
\node at (6.5,-3) {$F^\infty$};
\end{tikzpicture}
    \caption{Configuration $y^n$ with ${n=K=3}$.}\label{fig: 2 sig}
  \end{subfigure}
  \hfill
  \begin{subfigure}[b]{0.45\textwidth}
    \centering
\begin{tikzpicture}[thick,scale=0.4, every node/.style={scale=0.9}]

\draw  (0,0) rectangle (1,1);
\draw  (2,-4) rectangle (3,-3);
\draw  (1,-4) rectangle (2,-3);
\draw  (4,-4) rectangle (5,-3);
\draw  (3,-4) rectangle (4,-3);
\draw  (0,-1) rectangle (1,0);
\draw  (0,-3) rectangle (1,-2);
\draw  (0,-2) rectangle (1,-1);
\draw  (0,-5) rectangle (1,-2);
\draw  (0,-4) rectangle (1,-3);
\draw  (0,-7) rectangle (1,-6);
\draw  (0,-6) rectangle (1,-5);
\node at (0.5,0.5) {$s_1$};
\node at (0.5,-0.5) {$w$};
\node at (0.5,-1.5) {$w$};
\node at (0.5,-2.5) {$w$};
\node at (0.5,-3.5) {$\triangle$};
\node at (0.5,-4.5) {$w$};
\node at (0.5,-5.5) {$w$};
\node at (0.5,-6.5) {$w$};

\node at (1.5,-3.5) {$w$};
\node at (2.5,-3.5) {$w$};
\node at (3.5,-3.5) {$w$};
\node at (4.5,-3.5) {$w$};
\draw  (8,0) rectangle (9,1);
\draw  (10,-4) rectangle (11,-3);
\draw  (9,-4) rectangle (10,-3);
\draw  (12,-4) rectangle (13,-3);
\draw  (11,-4) rectangle (12,-3);
\draw  (8,-1) rectangle (9,0);
\draw  (8,-3) rectangle (9,-2);
\draw  (8,-2) rectangle (9,-1);
\draw  (8,-5) rectangle (9,-2);
\draw  (8,-4) rectangle (9,-3);
\draw  (8,-7) rectangle (9,-6);
\draw  (8,-6) rectangle (9,-5);
\node at (8.5,0.5) {$s_1$};
\node at (8.5,-0.5) {$s_1$};
\node at (8.5,-1.5) {$s_1$};
\node at (8.5,-2.5) {$s_1$};
\node at (8.5,-3.5) {$s_2$};
\node at (8.5,-4.5) {$s_2$};
\node at (8.5,-5.5) {$s_2$};
\node at (8.5,-6.5) {$s_2$};

\node at (9.5,-3.5) {$w$};
\node at (10.5,-3.5) {$w$};
\node at (11.5,-3.5) {$w$};
\node at (12.5,-3.5) {$w$};
\draw[->] (5.5,-3.5) -- (7.5,-3.5);
\node at (6.5,-3) {$F^\infty$};
\end{tikzpicture}
    \caption{Configuration $y^\infty$ with ${K=3}$.}\label{fig: 1 sig}
  \end{subfigure}%
  \caption{Configurations $y^n$ and $y^\infty$ and the limit fixed point reached under $F$.}\label{fig:confsy}
\end{figure}
      \[x^n(\vec i) =
      \begin{cases}
        s_2&\text{ if }\vec i=(0,b)\text{ with }b>-n\\
        s_1&\text{ if }\vec i=(0,b)\text{ with }b\leq -n\\
        0&\text{ else.}
      \end{cases}
      \]
      \[y^n(\vec i) =
      \begin{cases}
        s_1&\text{ if }\vec i=(0,b)\text{ with }b>K\\
        w&\text{ if }\vec i=(0,b)\text{ with }b\leq K\text{ and }b\neq 0\\
        \triangle&\text{ if }\vec i=(0,0)\\
        s_1&\text{ if }\vec i=(a,0)\text{ with }a>n\\
        w&\text{ if }\vec i=(a,0)\text{ with }a\leq n\text{ and }a>0\\
        0&\text{ else.}
      \end{cases}
      \]
  where $K$ is a large enough constant (compared to $E$) and $x^\infty$ (resp. $y^\infty$) is the limit of $x^n$ (resp. $y^n$) when $n$ goes to $\infty$. Choosing ${n=K}$ large enough (compared to $E$) it is straightforward to check that $x^\infty$ (resp. $y^n$) is $E$-locally reachable from $x^n$ (resp. $y^\infty$). From Lemma~\ref{lem:localreach}, we deduce that ${\overline\phi(x^\infty)\prec\overline\phi(x^n)}$ and ${\overline\phi(y^n)\prec\overline\phi(y^\infty)}$ where $\overline\phi$ is the encoding map involved in the simulation of $F$ by $G$. Thus we also have ${G^\infty\bigl(\overline\phi(y^n)\bigr)\prec G^\infty\bigl(\overline\phi(y^\infty)\bigr)}$ (by monotonicity of $G$), which is equivalent to ${\overline\phi(F^\infty(y^n))\prec\overline\phi(F^\infty(y^\infty))}$ by the simulation. Denoting by $P$ the half-plane of all positions ${(a,b)}$ with ${b<-n}$, we have ${F^\infty(y^\infty)_{|P} = x^\infty_{|P}}$ and ${F^\infty(y^n)_{|P} = x^n_{|P}}$. This translates by the simulation into ${\overline\phi(F^\infty(y^\infty))_{|P'} = \overline\phi(x^\infty)_{|P'}}$ and ${\overline\phi(F^\infty(y^n))_{|P'} = \overline\phi(x^n)_{|P'}}$ where $P'$ is some half-plane contained in $P$ (depending on parameters of the simulation). We reached a contradiction because the left-hand terms and the right-hand terms of this pair of equality compare differently with respect to $\prec$ as established above, and they cannot be all equal because ${\overline\phi(x^n)}$ is distinct from ${\overline\phi(x^\infty)}$ over $P'$ by injectivity of $\overline\phi$.
\end{proof}

\end{document}